\documentclass[sigconf]{acmart}
\AtBeginDocument{%
  }

\usepackage{enumitem}
\usepackage{multirow}
\usepackage{booktabs}
\usepackage{graphicx}
\usepackage{threeparttable}
\usepackage{makecell}
\usepackage{subcaption}
\usepackage{tabularx}
\usepackage{tabularray}
\UseTblrLibrary{booktabs}
\usepackage{hyperref}
\usepackage{algorithm}
\usepackage{algpseudocode}
\usepackage{balance} 
\makeatletter
\newcommand{\AlgInput}[1]{\Statex \hspace*{-\algorithmicindent}\textbf{Input:} #1}
\newcommand{\AlgI}[1]{\Statex \hspace*{-\algorithmicindent}{}#1}
\newcommand{\AlgOutput}[1]{\Statex \hspace*{-\algorithmicindent}\textbf{Output:} #1}
\makeatother
\usepackage{float}

\usepackage{amsmath, amsthm, amssymb}

\allowdisplaybreaks[2]  
\newtheorem{theorem}{\textbf{Theorem}}
\usepackage[labelformat=parens]{subcaption}  
\usepackage{xcolor}
\usepackage{colortbl}

\definecolor{orange1}{RGB}{255,245,235} 
\definecolor{orange2}{RGB}{255,230,200}
\definecolor{orange3}{RGB}{255,210,160}
\definecolor{orange4}{RGB}{255,180,120}
\definecolor{orange5}{RGB}{255,140,100} 
\definecolor{blue1}{RGB}{235,245,255}
\definecolor{blue2}{RGB}{200,230,255}
\definecolor{blue3}{RGB}{160,210,255}
\definecolor{blue4}{RGB}{120,180,255}
\definecolor{blue5}{RGB}{70,180,255}

\copyrightyear{2026}
\acmYear{2026}
\setcopyright{cc}
\setcctype{by-nc-nd}
\acmConference[KDD '26]{Proceedings of the 32nd ACM SIGKDD Conference on Knowledge Discovery and Data Mining V.2}{August 09--13, 2026}{Jeju Island, Republic of Korea}
\acmBooktitle{Proceedings of the 32nd ACM SIGKDD Conference on Knowledge Discovery and Data Mining V.2 (KDD '26), August 09--13, 2026, Jeju Island, Republic of Korea}
\acmDOI{10.1145/3770855.3817704}
\acmISBN{979-8-4007-2259-2/2026/08}

\author{Yajie Yu}
\affiliation{%
  \department{School of Computer Science and Information Security}
  \institution{Guilin University of Electronic Technology}
  \city{Guilin}
  \state{Guangxi}
  \country{China}}
\email{yyj@mails.guet.edu.cn}

\author{Chenzhong Bin}
\authornote{Corresponding author}
\affiliation{%
  \department{School of Computer Science and Information Security}
  \institution{Guilin University of Electronic Technology}
  \city{Guilin}
  \state{Guangxi}
  \country{China}}
\email{binchenzhong@guet.edu.cn}

\author{Zhoubo Xu}
\affiliation{%
  \department{School of Computer Science and Information Security}
  \institution{Guilin University of Electronic Technology}
  \city{Guilin}
  \state{Guangxi}
  \country{China}}
\email{xzbli_11@guet.edu.cn}

\author{Zhixin Zeng}
\affiliation{%
  \department{School of Computer Science and Information Security}
  \institution{Guilin University of Electronic Technology}
  \city{Guilin}
  \state{Guangxi}
  \country{China}}
\email{zxzeng@guet.edu.cn}

\author{Tongxin Xu}
\affiliation{%
  \department{School of Computer Science and Information Security}
  \institution{Guilin University of Electronic Technology}
  \city{Guilin}
  \state{Guangxi}
  \country{China}}
\email{xutoncy@gmail.com}

\author{Cihan Xiao}
\affiliation{%
  \department{Electrical and Computer Engineering}
  \institution{Johns Hopkins University}
  \city{Baltimore}
  \state{MD}
  \country{United States}}
\email{cxiao7@jhu.edu}

\author{Jiafeng Wu}
\affiliation{%
  \department{School of Computer Science and Information Security}
  \institution{Guilin University of Electronic Technology}
  \city{Guilin}
  \state{Guangxi}
  \country{China}}
\email{wujiafeng@mails.guet.edu.cn}

\renewcommand{\shortauthors}{Yajie Yu et al.}

\graphicspath{{./figs/}}

\begin{document}

\title{Beyond Instance-Level Alignment and Uniformity:\\ Semantic Factor Learning for Collaborative Filtering}

\begin{abstract}
Collaborative filtering (CF) is widely used in recommender systems (RecSys) due to its simplicity and efficiency. However, existing CF methods follow an instance-level learning paradigm. During the instance learning stage, a large number of uninteracted user-item instances, of which items are potential interested by the user, are incorrectly treated as true negative samples resulting in a severe limitation to the  generalization and scalability of models. Moreover, mainstream graph convolutional networks (GCNs) inherently suffer from high computational cost and over-smoothing issues, which limit the ability in capturing higher-order connectivity and lead to a poor generalization under sparse supervision signals. To address the above limitations, we propose \textbf{\underline{S}}em\textbf{\underline{a}}ntic \textbf{\underline{F}}actor \textbf{\underline{e}}nhanced \textbf{\underline{A}}lignment and \textbf{\underline{U}}niformity (\textbf{SaFeAU}), a novel framework that augments interacted instances with semantic factors, thereby mitigating false negative labeling and enabling matrix factorization (MF) to capture high-order CF signals without graph neighborhood aggregation. Specifically, SaFeAU consists of three tightly coupled components. First, Semantic Factor Routing (SFR) disentangles item representations into independent and global semantic factors. Building on these factors, Semantic Factor Matching (SFM) identifies uninteracted items, which share the same semantic factors with interacted ones, as potential positive pairs for enriching sparse supervision signals. Finally, Semantic Pairs Alignment (SPA) aligns both observed and potential positive pairs while promoting uniformity of user and item representations. Extensive experiments on four sparse real-world datasets show that SaFeAU consistently outperforms GCN-based and MF-based state-of-the-art CF methods in both recommendation accuracy and computational efficiency, confirming the effectiveness of the proposed semantic enhanced learning paradigm. The source code is available \footnote{\url{https://github.com/MysticRealm/SaFeAU}}.

\end{abstract}

\ccsdesc[500]{Information systems~Recommender systems}

\keywords{Recommender Systems, Collaborative Filtering, Semantic Factor Learning, Semantic Pairs Alignment}

\maketitle

\section{Introduction}
Recommender systems (RecSys) play a vital role in modern web services by helping users manage excessive information and discover personalized content such as products, videos, and other relevant items. Owing to its simplicity and effectiveness, collaborative filtering (CF) has become a widely adopted technique in RecSys for helping users discover potentially relevant items \cite{LightGCN,NGCF}. When only implicit feedback is available, prior studies \cite{SimpleX,DirectAU,effectiveness,incorporating,MAWU} emphasize that the effectiveness of CF models hinges on how supervision signals, e.g., user-item interactions, are accurately constructed to reflect true user preferences. CF methods can be broadly categorized into matrix factorization (MF) and graph convolutional networks (GCNs).

However, both types of CF methods \cite{BPR,LightGCN,DirectAU,LightGODE,LightCCF,FourierKAN-GCF,KDD26} naturally belong an \textbf{instance-level} learning paradigm, which aims at modeling CF signals by aligning interacted user-item instances, i.e. observed positive pairs, while pushing apart uninteracted user-item instances, i.e., unobserved negative pairs. Obviously, for a specific user, these negative pairs could contain a large number of potential positive pairs, a.k.a. false negative items, with which the user uninteracted but probably would be interested in. Indiscriminately pushing apart all of unobserved negative pairs will sacrifice the generalization of CF methods under large scale recommending scenarios. 

\vspace{-0.5em}
\begin{figure}[ht]
    \centering
    \includegraphics[width=1.02\linewidth]{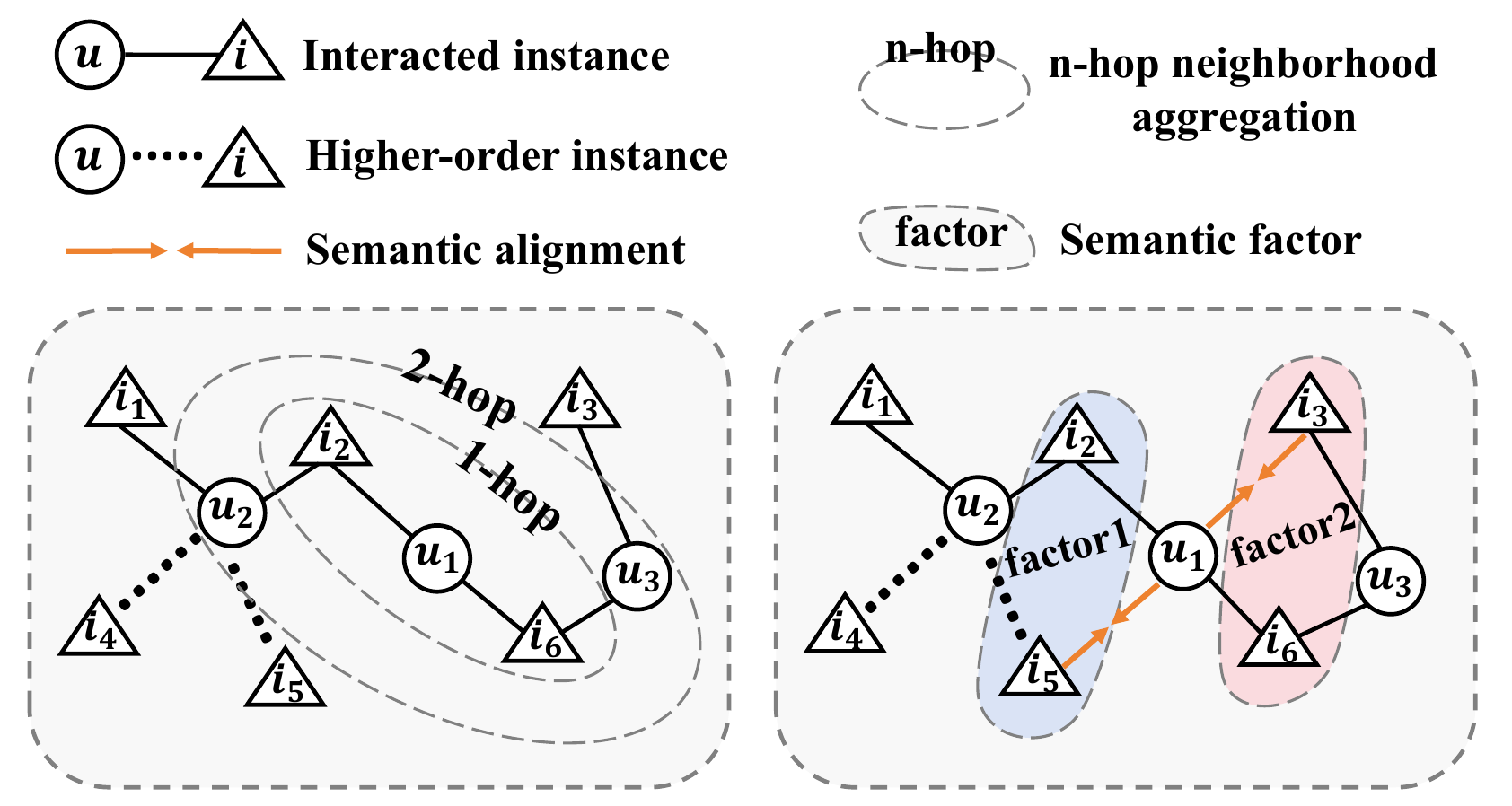} \\
    \vspace{-1.5em}
    
    \subcaptionbox{\label{fig1:a}}[0.48\linewidth]{\centering\hspace{0.5\linewidth}}%
    \hfill%
    \subcaptionbox{\label{fig1:b}}[0.48\linewidth]{\raggedright\hspace{0.5\linewidth}}
    \vspace{-0.6em}
    \caption{A comparison between (a) GCN-based alignment and (b) semantic positive pairs alignment on the same interaction graph. \label{fig:1}}
\end{figure}
\vspace{-0.5em}

Moreover, most GCN-based methods \cite{NGCF, LightGCN, GraphAU, KUPN, KMCLR, PKEF} harness neighborhood aggregation mechansim to learn user/item representations in interaction garphs for recommendations, inherently suffering from the high computational cost and the over-smoothing issue \cite{oversmoothing,over-smoothing}. Meanwhile, they are merely trained by interacted instances, thereby limiting the ability of capturing higher-order connectivity \cite{UltraGCN,SGL}, leading to incomplete user representations and poor generalization under sparse supervision signals \cite{data-sparsity}. For instance, as illustrated in Figure \ref{fig1:a}, user $u_1$ only directly interacted with items $i_2$ and $i_6$. As the multi-hop neighborhood aggregation proceeding, several irrelevant nodes are introduced as noise, unexpectedly contaminating the representation of user $u_1$. Meanwhile, higher-order neighbors such as $i_4$ and $i_5$ fail to establish effective connections with $u_1$ due to the multi-hop neighborhood relationships. Consequently, under sparse supervision signals, GCN-based methods struggle to accurately model fine-grained user preferences.

Naturally, the key to addressing the above two limitations lies in: \textit{How to identify false negative items based on sparse supervision signals? And how to achieve efficient semantic-level signal alignment to bypass the noisy neighborhood aggregation?}

In this paper, we propose \textbf{\underline{S}}em\textbf{\underline{a}}ntic \textbf{\underline{F}}actor \textbf{\underline{e}}nhanced \textbf{\underline{A}}lignment and \textbf{\underline{U}}niformity (\textbf{SaFeAU}), a novel framework that identifies false negative items and aligns user-item representations at the semantic-level through global semantic factor learning. The goal of our paper is to learn potential item attributes as a set of distinct semantic factors, related to, e.g., brand, material, texture, and appearance, based on the entire item set. Then, we identify false negative items among unobserved items which share the similar semantic factors with the user’s interacted items. As shown in Figure \ref{fig1:b}, after calculating the relationship between items and their corresponding semantic factors, items $i_2$ and $i_5$ are assigned to semantic factor1, while items $i_3$ and $i_6$ belong to semantic factor2. Despite merely interacted with $i_2$ and $i_6$, user $u_1$ is expected to align with potential positive pairs $i_3$ and $i_5$ as sharing the same semantic factors. This semantic positive pairs alignment overcomes the limitation of the multi-hop aggregation of GCNs, enabling user-item alignment in the semantic space via potential positive pair matching, thereby capturing a more complete and accurate representation of user preferences.

To achieve the above goal, specifically, SaFeAU consists of three tightly coupled components. First, a Semantic Factor Routing (SFR) algorithm is employed to learn a set of shared semantic factors through an iterative routing process. By iteratively updating semantic assignment weights, the model refines how items correlate with these factors, thereby encouraging disentanglement of multifaceted item attributes. Semantic Factor Matching (SFM) is then introduced to augment the sparse supervision signals of CF from a semantic perspective. Specifically, for each item interacted by a given user, we select its top-$k$ semantic factors based on semantic assignment weights and construct potential positive pairs by correlating the user with unobserved items that share any of these factors, thereby expanding the supervision signals based on observed positive pairs. As illustrated in Figure \ref{fig1:b}, user $u_1$ is inferred to have a semantic correlation with a high-order neighbor item $i_5$, since $i_5$ shares the same factor1 with the interacted item $i_2$. Finally, Semantic Pairs Alignment (SPA) leverages the alignment and uniformity principle to align users with both their interacted items and also with SFM-identified potential positive items, while promoting uniformity across all user and item representations for enhancing their disparity. Benefiting from the above components, SaFeAU elaborately expands the supervision signals of CF, improving the model generalization and training efficiency in the sparse and large-scale recommendation scenarios. The main contributions of this paper can be summarized as follows: 

\begin{itemize}[leftmargin=*]
\item To the best of our knowledge, this is the first work to explore false negative items from a semantic factor perspective. We propose a framework SaFeAU to introduce a novel semantic enhanced learning paradigm for CF, which augments interacted instances with semantic factors, thereby mitigating false negative labeling and enabling MF to capture high-order CF signals avoiding graph neighborhood aggregation.

\item We propose the Semantic Factor Routing (SFR) algorithm to learn shared semantic factors from item representations via semantic assignment weights updates, promoting disentanglement of multifaceted item attributes. Building upon these semantic factors, we further introduce the Semantic Factor Matching (SFM) method to enrich sparse supervision signals in CF by constructing potential positive pairs through matching each user with unobserved items that share semantic factors with their interacted items.

\item Extensive experiments on four sparse real-world datasets solidly demonstrate that SaFeAU consistently outperforms GCN-based and MF-based state-of-the-art CF methods in both recommendation accuracy and computational efficiency, validating the effectiveness of the semantic enhanced learning paradigm.
\end{itemize}

\section{Preliminaries}

In this section, we first formalize the CF problem, then theoretically analyze the relationships among three mainstream CF loss functions, and finally theoretically justify the advantages of the alignment and uniformity (AU) optimizations.

\subsection{Problem Definition}
Let $\mathcal{U}$ and $\mathcal{I}$ denote the sets of users and items, respectively. The data distribution in CF is characterized by the marginal distributions $p_{\mathrm{user}}(u)$ and $p_{\mathrm{item}}(i)$ over $\mathcal{U}$ and $\mathcal{I}$, as well as the joint distribution $p_{\mathrm{pos}}(u, i)$ of observed positive user--item interactions.

Most CF methods employ an encoder network $f(\cdot)$ that maps each user and item into a $d$-dimensional embedding space $\mathbb{R}^d$. $\tilde{f(\cdot)}$ means normalized representations. Then, the predicted score is usually defined as dot product $s(u,i) = \tilde{f(u)}^\top \tilde{f(i)}$. For simplicity, the Euclidean distance between user $u$ and item $i$ is denoted as $d(u,i) = ||\tilde{f(u)}-\tilde{f(i)}||^2$. The objective of CF is to utilize the interactive instances to predict preference scores for unobserved user-item pairs and rank items to generate recommendations.

\begin{figure*}[htbp] \centering \includegraphics[width=\textwidth]{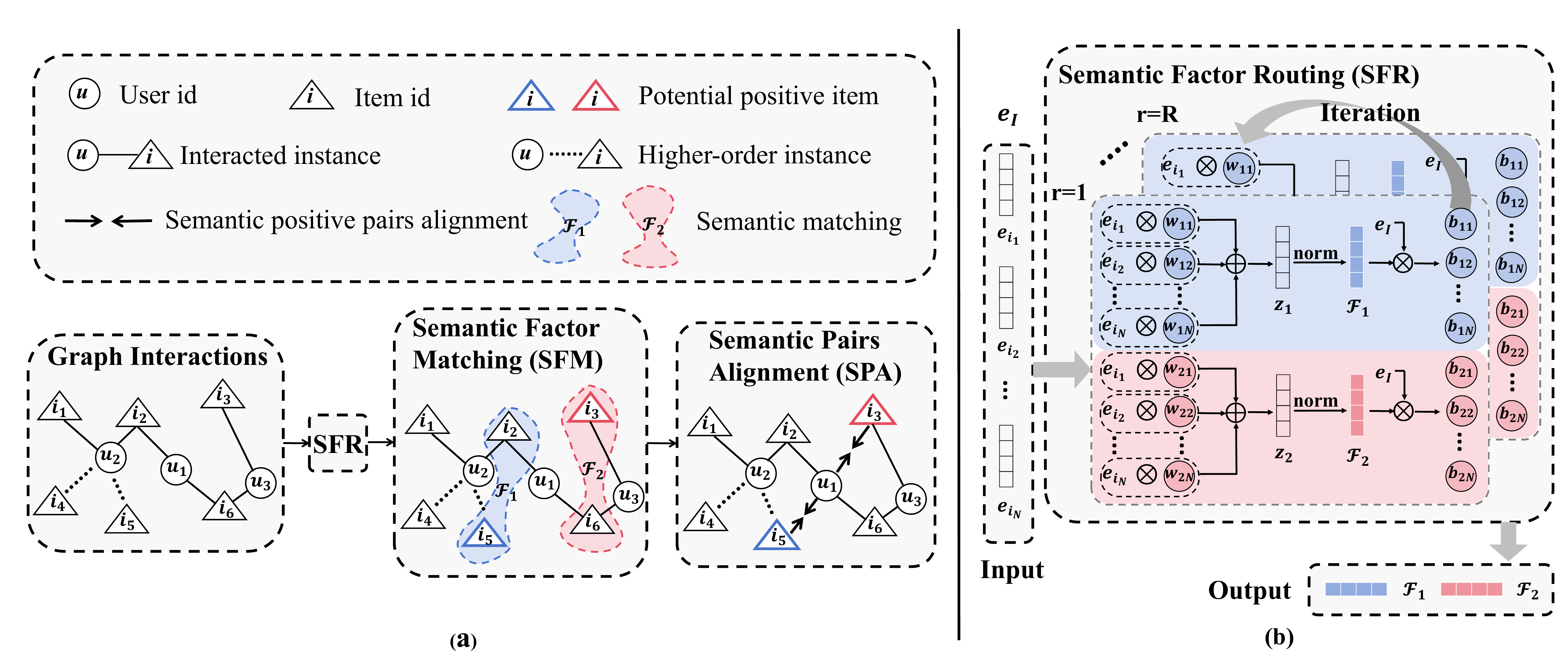} \caption{The architecture of SaFeAU. (a) Taking user $u_1$ as an example, we show how to expand its potential positive items. (b) Illustration of the Semantic Factor Routing (SFR) algorithm, which performs $r$ iterations to route all item embeddings to semantic factors based on the semantic assignment weights, resulting in disentangled factor embeddings. } \label{fig} \end{figure*}

\subsection{BPR and InfoNCE for CF}
The BPR objective \cite{BPR} discriminates the similarity of the positive pair $(u,i)$ from that of the negative pair $(u,i^-)$:

\begin{equation}
\begin{aligned}
L_{BPR} = - \underset{(u,i) \sim p_{\text{pos}}}{\mathbb{E}} \log \left[ \text{sigmoid} \left( s(u,i) - s(u,i^-) \right) \right].
\label{1}
\end{aligned}
\end{equation}

Recently, the study of LightCCF \cite{LightCCF} has demonstrated that the InfoNCE \cite{InfoNCE} loss can effectively model user-item interactions, which incorporates a set of $M$ negative pairs $\{j_n\}_{n=1}^{M}$: 

\begin{equation}
\begin{aligned}
L_{\text{CL}}(f; \tau, M) = - \underset{\substack{(u,i) \sim p_{\text{pos}} \\ \{j_n\}_{n=1}^{M} \overset{\text{i.i.d.}}{\sim} p_{\text{item}}}}{\mathbb{E}} \log \frac{e^{\text{s}(u,i)/\tau}}{e^{\text{s}(u,i)/\tau} + \sum_{n=1}^{M} e^{\text{s}(u,j_n)/\tau}},
\label{2}
\end{aligned}
\end{equation}
where $\tau$ is the temperature hyperparameter. When $M=1$ and $\tau=1$, the InfoNCE objective becomes mathematically equivalent to BPR. 

Indeed, prior work has shown that a larger $M$ provides broader coverage of the item distribution, which yields a more informative contrastive signal \cite{InfoNCE}. In Appendix \ref{M}, we theoretically prove that InfoNCE outperforms BPR due to its enhanced negative sampling coverage, which leads to a more accurate approximation of the underlying data distribution.

\subsection{Alignment and Uniformity Optimizations Bypass Negative Sampling Limits}
\label{bypassing}
DirectAU \cite{DirectAU} builds upon the insight drew from \cite{AU} that representation quality is governed by two key properties: alignment and uniformity (AU). Alignment is straightforwardly defined to enhance the proximity between normalized embeddings of positive pairs under $p_{\mathrm{pos}}(u, i)$:

\begin{equation}
\begin{aligned}
L_{\text{align}} = \underset{(u,i) \sim p_{\text{pos}}}{\mathbb{E}} \left\| \tilde{f(u)} - \tilde{f(i)} \right\|^2.
\label{align}
\end{aligned}
\end{equation}

While the uniformity loss is defined as the logarithm of the average pairwise Gaussian potential:
\begin{equation}
\begin{aligned}
L_{\text{uniform}} = \frac{1}{2} \log \underset{(u,u') \overset{\text{i.i.d.}}{\sim} p_{\text{user}}}{\mathbb{E}}  e^{-2 \left\| \tilde{f(u)} - \tilde{f(u')} \right\|^2} +\\
\frac{1}{2} \log \underset{(i,i') \overset{\text{i.i.d.}}{\sim} p_{\text{item}}}{\mathbb{E}}  e^{-2 \left\| \tilde{f(i)} - \tilde{f(i')} \right\|^2},
\label{uniform}
\end{aligned}
\end{equation}
where $(u,u') \sim p_{\text{user}}$ and $(i, i') \sim p_{\text{item}}$ are pairs of users and items sampled independently and identically from their respective distributions, i.e., i.i.d. These two losses directly reflect the fundamental principle of contrastive representation learning: encouraging learned feature representation for positive pairs to be similar, while pushing the randomly sampled negative pairs apart \cite{AU}.

We analyze the relationship between AU and InfoNCE, revealing that AU avoids negative sampling by enforcing uniformity over the entire set of user and item embeddings, thereby eliminating performance fluctuations derived from the random negative sampling. The proof is detailed in Appendix \ref{M}. However, under sparse supervision signals, the effectiveness of AU principle is limited due to the sparse observed positive pairs and fails to further exploit potential positive pairs. This inherent weakness motivates us to study how to enrich supervision signals from a semantic enhanced perspective, thereby strengthening the generalization of CF models. 

\section{Methodology}
In this section, we propose the \textbf{\underline{S}}em\textbf{\underline{a}}ntic \textbf{\underline{F}}actor \textbf{\underline{e}}nhanced \textbf{\underline{A}}lignment and \textbf{\underline{U}}niformity framework (\textbf{SaFeAU}), which breaks through the limitations of conventional CF learning paradigm that operates solely on the instance-level. By introducing semantic factors to identify potential positive pairs, SaFeAU elaborates fine-grained supervision signals without the multi-hop neighborhood propagation, thereby avoiding the computational complexity and over-smoothing issues of GCNs.

\subsection{Semantic Factor Routing (SFR)}

\begin{algorithm}[htbp]
\caption{Semantic Factor Routing (SFR)}
\label{alg:SFR}
\begin{algorithmic}[1]
\AlgInput{item embeddings $\{\tilde{f(i)} \mid i \in \mathcal{I}_b\}$, iteration times $r$, number}
\AlgI{of semantic factors $K$}
\AlgOutput{semantic factors $\{\mathcal{F}_j \mid j = 1, \dots, K\}$}
\State Initialize routing logits: $b_{ij} \sim \mathcal{N}(0, \sigma^2)$ for all item $i$ and semantic factor $j$.
\For{$t = 1$ to $r$}
    \State Compute semantic assignment weights:
    \Statex \quad $w_{ij} \gets \mathrm{softmax}(b_{ij})$
    \State Aggregate semantic factor: $z_j \gets \sum_{i \in \mathcal{I}_b} w_{ij} \tilde{f(i)}$
    \State Normalize: $\mathcal{F}_j \gets \mathrm{normalize}(z_j)$
    \State Update routing logits: $b_{ij} \gets b_{ij} + \mathcal{F}_j^\top \tilde{f(i)}$
\EndFor
\State \Return $\{\mathcal{F}_j\}_{j=1}^K$
\end{algorithmic}
\end{algorithm}

Following the batch-wise training framework of DirectAU \cite{DirectAU}, semantic factors are computed using only the items in the current training batch $\mathcal{I}_b$. This batch-wise computation significantly reduces memory usage and computational cost, while also making it more consistent with the actual data distribution of users and items. To learn semantic factors within each batch, we propose a Semantic Factor Routing (SFR) algorithm, inspired by the dynamic routing procedure in capsule networks \cite{Dynamic-routing}. Specifically, SFR dynamically routes each normalized item representation $\tilde{f(i)} \in \mathbb{R}^d$ ($i \in \mathcal{I}_b$) to $K$ high-level semantic factors $\{\mathcal{F}_j \in \mathbb{R}^d\}_{j=1}^K$, where each factor captures a mutually independent attribute such as brand or material. This shared embedding dimensionality $d$ ensures that item representations and semantic factors lie in the same vector space, allowing direct computation without intermediate transformation matrices.

The semantic assignment weight $w_{ij}$ determines how much item $i$ contributes to semantic factor $j$. To enable adaptive assignment of items to semantic factors, we first introduce learnable routing logits $b_{ij}$ that govern the soft routing of item $i$ to semantic factor $j$. These logits are initialized with small random values, i.e., $ b_{ij} \sim \mathcal{N}(0, \sigma^2)$, providing an unbiased normal initialization for the subsequent iterative routing process. In each routing round, the semantic assignment weight $w_{ij}$ is updated by applying a softmax normalization over the routing logits $b_{ij}$:

\begin{equation}
w_{ij} = \frac{\exp b_{ij}}{\sum_{k=1}^K \exp b_{ik}}.
\end{equation}

The intermediate representation $z_j$ is obtained as the weighted sum of all item embeddings in the current training batch $\mathcal{I}_b$:
\begin{equation}
    z_j = \sum_{i \in \mathcal{I}_b} w_{ij} \tilde{f(i)}.
\end{equation}

This representation is then normalized to yield the semantic factor:
\begin{equation}
    \mathcal{F}_j = \frac{z_j}{\|z_j\|}.
\end{equation}

Finally, the routing logits are updated to reinforce the agreement between items and factors:
\begin{equation}
    b_{ij} \leftarrow b_{ij} + \mathcal{F}_j^\top \tilde{f(i)}.
\end{equation}
These steps are repeated for $r$ rounds to iteratively refine the semantic factors. The complete SFR algorithm is detailed in Algorithm \ref{alg:SFR}.

\begin{table*}
\centering
\caption{Time complexity comparison in SaFeAU with the current state-of-the-art methods.}
\label{tab:time_complexity}
\setlength{\tabcolsep}{6pt} 
\renewcommand{\arraystretch}{1.0}
\begin{tabular}{c|c|c|c|c|c|c}
\toprule
\textbf{Components} & \textbf{LightGCN} & \textbf{FourierKAN-GCF} & \textbf{SimGCF} & \textbf{LightGODE} & \textbf{LightCCF} & \textbf{SaFeAU} \\
\midrule
Graph Construction & $\mathcal{O}(2|\mathcal{E}|)$ & $\mathcal{O}(2|\mathcal{E}|)$ & $\mathcal{O}(2|\mathcal{E}|)$ & -- & -- & -- \\
Graph Convolution & $\mathcal{O}(2|\mathcal{E}|\mathcal{L}d)$ & $\mathcal{O}(2|\mathcal{E}|\mathcal{L}d(1+g))$ & $\mathcal{O}(2n|\mathcal{E}|d)$ & -- & -- & -- \\
Loss Computation & $\mathcal{O}(2\mathcal{B}d)$ & $\mathcal{O}(2\mathcal{B}d)$ & $\mathcal{O}(2\mathcal{B}d)$ & $\mathcal{O}(\mathcal{B}d+2\mathcal{B}^2d)$ & $\mathcal{O}(3\mathcal{B}d+2\mathcal{B}^2d)$ & $\mathcal{O}(\mathcal{B}d+rK\mathcal{B}d+2\mathcal{B}^2d)$ \\
\bottomrule
\end{tabular}
\end{table*}

\subsection{Semantic Factor Matching (SFM)}

After obtained semantic factors $\{\mathcal{F}_j\}_{j=1}^K$, we propose Semantic Factor Matching (SFM) method to semantically correlate each item with its similar items under the shared semantic factors. For each item $i$, SFM selects the top-$k$ semantic factors based on the semantic assignment weights $\{w_{ij}\}_{j=1}^K$, where the selected factor indices set $\mathcal{K}_i$ is denoted as:

\begin{equation}
    \mathcal{K}_i = \operatorname{top}-k \left( \{ w_{ij} \}_{j=1}^K \right).
\end{equation}

To dynamically control the degree of semantic correlation between two items, we introduce a semantic threshold $\delta$ that defines the minimum number of shared top-$k$ semantic factors required for two items to judge whether a semantic correlation existing between them. Accordingly, for each interacted item $i$ ($i \in \mathcal{I}_b$), we define its semantic augmentation set $\mathcal{N}_i$ as follows:

\begin{equation}
    \mathcal{N}_i = \bigl\{ i' \in \mathcal{I}_b \setminus \{i\} \,\big|\, \lvert \mathcal{K}_i \cap \mathcal{K}_{i'} \rvert \geq \delta \bigr\}.
\label{threshold}
\end{equation}

Next, SFM augments the supervision signals beyond observed positive pairs by constructing potential positive pairs from the semantic augmentation set $\mathcal{N}_i$. Specifically, for each interacted instance $(u, i)$, any item $i' \in \mathcal{N}_i$ is regarded as a semantically consistent extension of $i$, yielding the augmented positive pair set $\{(u, i') \mid i' \in \mathcal{N}_i\}$ for SPA module.

\subsection{Semantic Pairs Alignment (SPA)}

As discussed in Section \ref{bypassing}, AU abandons the explicit negative sampling through global uniformity and thus avoids the performance fluctuations caused by randomly sampled negative pairs. However, its alignment loss is confined to observed interactions. To extend supervision signals to semantically correlated but unobserved items, we propose Semantic Pairs Alignment (SPA) module, which constructs additional positive pairs based on shared semantic factors, thereby enhancing generalization under data sparsity.

Specifically, SPA mainly aims at aligning users with the semantic augmented items in $\mathcal{N}_i$ produced by SFM. The semantic positive pairs alignment loss is formulated as:

\begin{equation}
    L_{\text{align}}^{\text{semantic}} = \underset{(u,i) \sim p_{\text{pos}}}{\mathbb{E}} \left[ 
         \frac{1}{|\mathcal{N}_i|} \sum_{i' \in \mathcal{N}_i} \left\| \tilde{f(u)} -\tilde{f(i')} \right\|^2
    \right].
\end{equation}

By pulling user representations toward those of semantically consistent unobserved items, SPA effectively propagates supervision signals through shared semantic factors. This encourages the model to generalize user preferences to semantically correlated but unobserved items, thereby mitigating data sparsity.

Additionally, our SPA leverages the base alignment loss $L_{\text{align}}$ and uniformity loss $L_{\text{uniform}}$, which are defined in Equation \eqref{align} and Equation \eqref{uniform}, respectively, to effectively model the interacted instances as well.

\subsection{Model Optimization}

To learn high-quality user and item representations, we utilize a multi-task training strategy \cite{SGL} to jointly optimize the base alignment loss, the uniformity loss, and the semantic positive pairs alignment loss. The total loss is defined as:

\begin{equation}
    L_{\text{SaFeAU}} = L_{\text{align}}+ \gamma_1 L_{\text{uniform}}+ \gamma_2 L_{\text{align}}^{\text{semantic}},
\end{equation}
where $\gamma_1 > 0$ and $\gamma_2 > 0$ are hyperparameters controlling the relative importance of uniformity and semantic positive pairs alignment, respectively. After completing the model training process, we use the dot product to predict unknown preferences for recommendations.

\subsection{Time Complexity Analysis}
To demonstrate the efficiency of SaFeAU, we analyze its time complexity and compare it with that of GCN-based methods including LightGCN \cite{LightGCN}, FourierKAN-GCF \cite{FourierKAN-GCF}, and SimGCF \cite{KDD26}, as well as MF-based methods including LightGODE \cite{LightGODE} and LightCCF \cite{LightCCF}. The complexity details of all methods are shown in Table \ref{tab:time_complexity}. We define the number of edges in the user-item bipartite graph as $|\mathcal{E}|$. Let $\mathcal{L}$ represent the number of graph convolution layers, $d$ and $\mathcal{B}$ represent the size of embedding and the batch size, respectively. In addition, for FourierKAN-GCF, $g$ denotes the Fourier frequency grid size, and for SimGCF, $n$ represents the polynomial filter order. We add the semantic factor routing to generate semantic factors, where $r$ denotes the maximum number of iterations and $K$ is the number of semantic factors. On this basis, we can derive the following facts:
\begin{itemize}[leftmargin=*, nosep]
\item In the graph construction process, both GCN-based methods require normalization of the adjacency matrix. This step involves computing $2|\mathcal{E}|$ non-zero elements of the original adjacency matrix. On the contrary, MF-based alleviates the need for graph construction and adjacency matrix normalization in training.
\item In the graph convolution stage, GCN-based methods perform message passing over the interaction graph, resulting in a time complexity of $\mathcal{O}(2|\mathcal{E}|\mathcal{L} d)$. FourierKAN-GCF has a higher complexity of $\mathcal{O}(2|\mathcal{E}|\mathcal{L} d (1 + g))$ due to edge-wise Fourier KAN transformations, while SimGCF incurs $\mathcal{O}(2n|\mathcal{E}|d)$ due to its degree-$n$ spectral polynomial approximation. In contrast, MF-based methods avoid explicit graph propagation, eliminating the dependence on $|\mathcal{E}|$ and facilitating the deployment in large-scale RecSys.
\item In the loss computation stage, GCN-based methods employ the BPR loss, resulting in a complexity of $\mathcal{O}(2\mathcal{B}d)$. LightCCF further incorporates a neighborhood alignment loss based on InfoNCE, resulting in a total loss complexity of $\mathcal{O}(3\mathcal{B}d + 2\mathcal{B}^2 d)$. Both LightGODE and SaFeAU adopt the AU loss formulation, which incurs a complexity of $\mathcal{O}(\mathcal{B}d + 2\mathcal{B}^2 d)$. The additional linear term in SaFeAU comes from generating semantic factors, which costs $\mathcal{O}(rK\mathcal{B}d)$, where $r$ and $K$ are small constants. Notably, the AU loss in SaFeAU is sampling-free, unlike BPR, which relies on sampled negative pairs, and can be parallel computed and optimized by simple matrix multiplications on GPUs. Overall, by combining semantic factor generation with AU loss, SaFeAU maintains a competitive computational complexity. Due to space limitation, the computational efficiency of SaFeAU is detailed in Appendix \ref{efficiency}.
\end{itemize}

\begin{table*}
\centering
\caption{A comparison between SaFeAU and state-of-the-art baselines is provided. The best values are shown in bold, while the second-best values are underlined, results are averaged over five runs. R@k refers to Recall@k, and N@k refers to NDCG@k. Improv.\% indicates the relative improvement over the top-performing baseline. An asterisk (*) indicates significant improvements with a t-test $p< 0.05$ compared to the best baseline.}
\label{tab:3}
\small
\begin{threeparttable}[b]
\resizebox{\linewidth}{!}{
\begin{tabular}{cccccccccccc}
\toprule
\multicolumn{2}{c}{\textbf{Setting}} & \multicolumn{4}{c}{\textbf{GCN-based}} & \multicolumn{4}{c}{\textbf{MF-based}} & \multicolumn{2}{c}{\textbf{Ours}} \\
\cmidrule(lr){1-2} \cmidrule(lr){3-6} \cmidrule(lr){7-10} \cmidrule(lr){11-12}
 Dataset & Metric & LightGCN & GraphAU & FourierKAN-GCF &  SimGCF & BPR-MF & DirectAU & LightGODE & LightCCF & SaFeAU & Improv. \\
\midrule
\multirow{4}{*}{Gowalla} 
 & R@10 & 0.0887 & 0.0976 & 0.0956 & 0.0962 & 0.0490 & 0.0905 & \underline{0.0983} & 0.0965 & \textbf{0.1034}$^*$ &  \textbf{5.19\%} \\
 & N@10 & 0.0604 & 0.0673 & 0.0664 & 0.0667 & 0.0342 & 0.0627 & \underline{0.0680} & 0.0670 & \textbf{0.0713}$^*$  & \textbf{4.85\%} \\
 & R@20 & 0.1223 & 0.1416 & 0.1372 & 0.1379 & 0.0701 & 0.1313 & \underline{0.1428} & 0.1384 & \textbf{0.1505}$^*$ & \textbf{5.39\%} \\
 & N@20 & 0.0681 & 0.0801 & 0.0778 & 0.0783 & 0.0401 & 0.0742 & \underline{0.0803} & 0.0786 & \textbf{0.0844}$^*$ & \textbf{5.11\%} \\
\midrule
\multirow{4}{*}{Toys-and-Games} 
 & R@10 & 0.0918 & 0.0934 & 0.0942 & 0.0945 & 0.0782 & 0.0941 & \underline{0.0951} & 0.0950 & \textbf{0.1006}$^*$ & \textbf{5.78\%} \\
 & N@10 & 0.0512 & 0.0529 & 0.0534 & 0.0538 & 0.0434 & 0.0534 & 0.0540 & \underline{0.0552} & \textbf{0.0578}$^*$ & \textbf{4.71\%} \\
 & R@20 & 0.1129 & 0.1323 & 0.1315 & 0.1289 & 0.1126 & 0.1325 & \underline{0.1338} & 0.1301 & \textbf{0.1405}$^*$ & \textbf{5.01\%} \\
 & N@20 & 0.0551 & 0.0619 & 0.0628 & 0.0636 & 0.0522 & 0.0632 & 0.0639 & \underline{0.0641} & \textbf{0.0681}$^*$ & \textbf{6.24\%} \\
\midrule
\multirow{4}{*}{Beauty} 
 & R@10 & 0.0843 & 0.0992 & 0.0995 & 0.0998 & 0.0790 & 0.0983 & \underline{0.1014} & 0.0997 & \textbf{0.1066}$^*$  & \textbf{4.93\%} \\
 & N@10 & 0.0463 & 0.0561 & 0.0565 & 0.0561 & 0.0429 & 0.0559 & 0.0563 & \underline{0.0570} & \textbf{0.0599}$^*$  & \textbf{4.74\%} \\
 & R@20 & 0.1227 & 0.1427 & 0.1433 & 0.1429 & 0.1161 & 0.1390 & \underline{0.1451} & 0.1362 & \textbf{0.1509}$^*$  & \textbf{3.93\%} \\
 & N@20 & 0.0563 & 0.0673 & 0.0677 & 0.0675 & 0.0526 & 0.0665 & \underline{0.0681} & 0.0674 & \textbf{0.0715}$^*$  & \textbf{4.70\%} \\
\midrule
\multirow{4}{*}{Yelp2018} 
 & R@10 & 0.0495 & 0.0670 & 0.0662 & 0.0672 & 0.0381 & 0.0672 & 0.0672 & \underline{0.0676} & \textbf{0.0695}$^*$  & \textbf{2.81\%} \\
 & N@10 & 0.0393 & 0.0538 & 0.0536 & 0.0547 & 0.0300 & 0.0539 & 0.0540 & \underline{0.0548} & \textbf{0.0555}$^*$ & \textbf{1.28\%} \\
 & R@20 & 0.0814 & 0.1063 & 0.1058 & 0.1082 & 0.0639 & 0.1069 & 0.1076 & \underline{0.1099} & \textbf{0.1144}$^*$  & \textbf{4.09\%} \\
 & N@20 & 0.0505 & 0.0671 & 0.0669 & 0.0686 & 0.0387 & 0.0672 & 0.0676 & \underline{0.0691} & \textbf{0.0713}$^*$  & \textbf{2.61\%} \\
\bottomrule
\end{tabular}}
\end{threeparttable}
\end{table*}

\begin{table}
\centering
\caption{Statistics of datasets.}
\label{tab:1}
\begin{tabular}{ccccc}
\toprule
Dataset & Users & Items & Interactions & Density\\
\midrule
Gowalla & 64.1k & 164.5k & 2018.4k & 0.019\% \\
Toys-and-Games & 19.4k & 11.9k & 167.5k & 0.072\% \\
Beauty & 22.4k & 12.1k & 198.5k & 0.073\% \\
Yelp2018 & 31.7K & 38.0k & 1561.4k & 0.130\% \\
\bottomrule
\end{tabular}
\end{table}

\section{Experiments}
In this section, we conduct experiments to address the following questions:

\begin{enumerate}[label=Q\arabic*., leftmargin=*, itemsep=1pt]
    \item How does SaFeAU perform in CF compared to state-of-the-art baselines?
    
    \item How adaptation is the potential positive pair identifying componets of SaFeAU for various CF methods?
    
    \item How do different potential positive pair identifying strategies of SFM method affect the recommendation performance?
    
    \item Does SFM truly identify higher-quality potential positive pairs, beyond gains from random positive expansion?
    
    \item How does SaFeAU effectively align positive pairs compared to mainstream CF methods?
    
    \item How do different hyperparameter settings affect the recommendation performance of SaFeAU? 
    
\end{enumerate}

\subsection{Experimental Settings}

\subsubsection{\textbf{Datasets}} 

We evaluate SaFeAU on four real-world datasets: Gowalla \footnote{\url{https://snap.stanford.edu/data/loc-gowalla.html}}, Yelp2018 \footnote{\url{https://www.yelp.com/dataset}}, and the Amazon subsets \footnote{\url{https://jmcauley.ucsd.edu/data/amazon/links.html}} Toys-and-Games and Beauty. Following standard practice \cite{NGCF,LightGCN,DirectAU,bootstrapping,modeling}, we remove duplicate interactions and retain only users and items with at least five interactions. The resulting statistics are reported in Table \ref{tab:1}. For each user, we randomly split interactions into training, validation, and testing sets using an 80\%/10\%/10\% split. We compare SaFeAU against state-of-the-art MF-based and GCN-based CF methods. Additional details on baselines are provided in Appendix \ref{baselines}.

\subsubsection{\textbf{Evaluation Protocols}} 
\label{evalution}
To evaluate the performance of top-K recommendation, we employ Recall and Normalized Discounted Cumulative Gain (NDCG) as evaluation metrics. Recall@K measures how many target items are retrieved in a K-item recommended list, while NDCG@K further concerns about their positions in the list. To ensure the accuracy and reliability of our experimental results, we employ the full-ranking strategy \cite{evaluation} during testing, which involves ranking all uninteracted items for a given test user. Additionally, each experiment is repeated 5 times with different random seeds. We report the average scores across all metrics to ensure statistical significance.

\subsubsection{\textbf{Implementation Details}}
\label{implementation}
We use the RecBole \cite{RecBole} framework to implement all the methods for fair comparisons. Adam is used as the default optimizer and the maximum number of epochs is set to 300. Early stop is adopted if NDCG@20 on the validation dataset continues to drop for 10 epochs. We set the embedding size to 64 and the learning rate to 1e$-$3 for all the methods. The training batch size is set to 256 on Beauty and 1024 on the other three datasets. The weight decay is tuned among [0, 1e$-$8, 1e$-$6, 1e$-$4]. The default encoder $f$ in SaFeAU is a simple embedding table that maps user/item IDs to embeddings. The weight $\gamma_1$ of $L_{uniform}$ in SaFeAU is tuned within [0.1, 0.5, 2, 7, 10], and $\gamma_2$ of $L_{align}^{semantic}$ is tuned within [0.001, 0.01, 0.1, 10]. For hyperparameters that define the model architecture, we search $K$ over \{1, 2, 4, 7\}, $k$ over \{1, 2, 3, 4\}, and $r$ over \{1, 2, 3, 4, 5, 6, 7, 8\}. As for baseline-specific hyperparameters, we tune them in the ranges suggested by the original paper. All the parameters are initialized by xavier initialization.

\subsection{Recommendation Performance (Q1)}
Table \ref{tab:3} shows the performance comparison of different baseline methods and our SaFeAU. From the experimental results, we have the following observations:
\begin{itemize}[leftmargin=*, nosep]
\item Noticeably, SaFeAU achieves the highest scores in NDCG and Recall across all datasets, demonstrating its effectiveness in different recommendation scenarios. The pronounced performance improvements on sparse datasets, i.e., Gowalla and Toys, demonstrate SaFeAU’s robustness and generalization ability under high sparsity and large-scale scenarios.

\item Among all baselines, LightCCF and LightGODE emerge as the strongest competitors across all four datasets, demonstrating that after carefully designed in interaction modeling methods, MF-based approaches can outperform GCN-based models.

\item We observe that AU-based methods, including DirectAU, LightGODE, and GraphAU, perform better on sparser datasets. In contrast, LightCCF gains an advantage in denser datasets, which aligns with our analysis in Appendix \ref{M} showing that the AU principle avoids performance fluctuations caused by randomly sampled negative pairs.

\end{itemize}

\subsection{Compatibility with Existing CF Methods (Q2)}
\label{compatibility}

Apart from functioning as a stand-alone framework, the SFR and SFM components of SaFeAU can serve as a model-agnostic plugin for mainstream CF methods. We integrate it into BPR-MF, LightGCN, and LightCCF, and compare each variant with its original counterpart. As shown in Table \ref{tab:compatibility}, SaFeAU consistently achieves significant performance gains by refining supervision signals through semantically discovering potential positive pairs from uninteracted items.

\begin{table}[htbp]
\centering
\caption{Compatibility analysis with different CF models on three various density datasets. ``Rel Impr.'' indicates the relative improvements compared to the corresponding methods. \textcolor{orange5}{Orange} denotes an improvement, with darker shades indicating larger relative changes.}
\label{tab:compatibility}
\small
\setlength{\tabcolsep}{3.5pt} 
\begin{tabular*}{\columnwidth}{@{\extracolsep{\fill}}lcccccc@{}}
\toprule
\multirow{2}{*}{Method} & 
\multicolumn{2}{c}{Gowalla} & 
\multicolumn{2}{c}{Toys-and-Games} & 
\multicolumn{2}{c}{Yelp2018} \\
\cmidrule(lr){2-3} \cmidrule(lr){4-5} \cmidrule(lr){6-7}
 & R@20 & N@20 & R@20 & N@20 & R@20 & N@20 \\
\midrule
DirectAU & 0.1313 & 0.0742 & 0.1325 & 0.0632 & 0.1069 & 0.0672 \\
DirectAU+SaFe & 0.1505 & 0.0844 & 0.1405 & 0.0681 & 0.1144 & 0.0713 \\
\textbf{Rel Impr.} & \cellcolor{orange3} 14.62\% & \cellcolor{orange3} 13.75\% & \cellcolor{orange2} 6.04\% & \cellcolor{orange2} 7.75\% & \cellcolor{orange2} 7.02\% & \cellcolor{orange2} 6.10\% \\
\midrule
LightCCF & 0.1384 & 0.0786 & 0.1301 & 0.0641 & 0.1099 & 0.0691 \\
LightCCF+SaFe & 0.1422 & 0.0817 & 0.1306 & 0.0652 & 0.1131 & 0.0699 \\
\textbf{Rel Impr.} & \cellcolor{orange2} 2.75\% & \cellcolor{orange2} 3.94\% & \cellcolor{orange1} 0.38\% & \cellcolor{orange2} 1.72\% & \cellcolor{orange2} 2.91\% & \cellcolor{orange2} 1.16\% \\
\midrule
BPR-MF & 0.0701 & 0.0401 & 0.1126 & 0.0522 & 0.0639 & 0.0387 \\
BPR-MF+SaFe& 0.1163 & 0.0614 & 0.1129 & 0.0529 & 0.0778 & 0.0457 \\
\textbf{Rel Impr.} & \cellcolor{orange4} 65.91\% & \cellcolor{orange4} 53.12\% & \cellcolor{orange1} 0.27\% & \cellcolor{orange2} 1.34\% & \cellcolor{orange3} 21.75\% & \cellcolor{orange3} 18.09\% \\
\midrule
LightGCN & 0.1223 & 0.0681 & 0.1129 & 0.0551 & 0.0814 & 0.0505 \\
LightGCN+SaFe& 0.1317 & 0.0705 & 0.1135 & 0.0558 & 0.0877 & 0.0534 \\
\textbf{Rel Impr.} & \cellcolor{orange2} 7.69\% & \cellcolor{orange2} 3.52\% & \cellcolor{orange1} 0.53\% & \cellcolor{orange2} 1.27\% & \cellcolor{orange2} 7.74\% & \cellcolor{orange2} 5.74\% \\
\bottomrule
\end{tabular*}
\end{table}

\subsection{Potential Positive Pairs Identification (Q3)}
\label{SFM}

To evaluate the effectiveness of SFM in identifying potential positive pairs, we conduct an ablation study on its matching strategy. We compare four variants:

\begin{itemize}[leftmargin=*, nosep]
\item \textbf{None-SFM}: No unobserved items are identified as potential positives.
\item \textbf{Random-SFM}: Potential positives are randomly selected from unobserved items.
\item \textbf{Strict-SFM}: The unobserved items are considered as potential positives if they share all top-$k$ semantic factors with the given item, that is the semantic threshold of $\delta$ is set at k.
\item \textbf{SFM}: The unobserved items are considered as potential positives if they share at least $\delta$ semantic factors with the given item.
\end{itemize}

As shown in Table \ref{tab:ablation_sfa}, both Strict-SFM and SFM achieve consistent gains, validating the utility of potential positive pairs identification for the performance of CF. Notably, SFM outperforms Strict-SFM on all metrics. This is because a moderate semantic threshold such as $\delta = 2$ avoids insufficient coverage caused by high thresholds and the misidentification of true negatives as positives caused by low thresholds. Surprisingly, by comparing the results in Tables \ref{tab:3} and \ref{tab:ablation_sfa}, we find that Random-SFM outperforms all baselines on sparse datasets, i.e., Gowalla and Toys-and-Games, revealing that baseline methods fail to identify the widely existing false negatives. Conversely, on denser datasets, i.e., Beauty and Yelp2018 Random-SFM performs worse than None-SFM, underscoring the need for better mechanisms to identify potential positive items as data density grows.

\begin{table}[htbp]
\centering
\caption{An ablation study on SFM regarding the strategy of identifying potential positive pairs. For brevity, the dataset ``Toys-and-Games'' is abbreviated as Toys in the table. \textcolor{blue5}{Blue} denotes a decrease in performance relative to the None-SFM, and \textcolor{orange5}{Orange} denotes an improvement, with darker shades indicating larger relative changes.}
\label{tab:ablation_sfa}
\small 
\setlength{\tabcolsep}{4pt} 
\begin{tabular}{@{}llcccc@{}}
\toprule
Dataset & Metric & None-SFM & Random-SFM & Strict-SFM & SFM \\
\midrule
\multirow{2}{*}{Gowalla} & R@20 & 0.1315 & \cellcolor{orange2}0.1478 & \cellcolor{orange3}0.1501 & \cellcolor{orange3}\textbf{0.1505} \\
                         & N@20 & 0.0746 & \cellcolor{orange2}0.0828 & \cellcolor{orange3}0.0842 & \cellcolor{orange3}\textbf{0.0844} \\
\midrule
\addlinespace
\multirow{2}{*}{Toys} & R@20 & 0.1322 & \cellcolor{orange1}0.1381 & \cellcolor{orange2}0.1386 & \cellcolor{orange2}\textbf{0.1405} \\
                         & N@20 & 0.0627 & \cellcolor{orange1}0.0663 & \cellcolor{orange2}0.0675 & \cellcolor{orange2}\textbf{0.0681} \\
\midrule
\addlinespace
\multirow{2}{*}{Beauty} & R@20 & 0.1394 & \cellcolor{blue2}0.1295 & \cellcolor{orange3}0.1495 & \cellcolor{orange3}\textbf{0.1509} \\
                         & N@20 & 0.0667 & \cellcolor{blue2}0.0645 & \cellcolor{orange3}0.0711 & \cellcolor{orange3}\textbf{0.0715} \\
\midrule
\multirow{2}{*}{Yelp2018} & R@20 & 0.1062 & \cellcolor{blue1}0.1044 & \cellcolor{orange2}0.1132 & \cellcolor{orange2}\textbf{0.1144} \\
                         & N@20 & 0.0669 & \cellcolor{blue1}0.0649 & \cellcolor{orange2}0.0706 & \cellcolor{orange2}\textbf{0.0713} \\                         
\bottomrule
\end{tabular}
\end{table}

\subsection{Validation of Potential Positive Pair Quality (Q4)}

While SFM improves downstream ranking metrics, results from the ablation study in Table \ref{tab:ablation_sfa} show that Random-SFM can perform surprisingly well on sparse datasets. This raises the question of whether SFM's gains truly stem from precise semantic matching rather than general positive-expansion or relaxed-negative assumptions. To address this, we directly take the potential positive pairs identified by SFM on the training set for evaluation. Based on these selected pairs, we measure the proportion that match the ground-truth positives in the testing set using HR@10, NDCG@10, and Recall@10. The results are presented in Table \ref{tab:delta_analysis}.

SFM consistently outperforms Random-SFM under various settings of $\delta$, demonstrating that SFM identifies genuinely higher-quality potential positives rather than merely benefiting from random expansion. Furthermore, the optimal performance is consistently achieved at $\delta=2$, indicating that the mechanism of SFM is robust and does not require complex threshold tuning.

\begin{table}[htbp]
\centering
\small
\setlength{\tabcolsep}{0pt} 
\caption{Performance comparison of SFM under different semantic thresholds $\delta$, with top-$k$ semantic factors fixed at 3. H@10 denotes HR@10.}
\label{tab:delta_analysis}

\begin{tabular*}{\columnwidth}{@{\extracolsep{\fill}} l ccc ccc @{}} 
\toprule
\multirow{2}{*}{\textbf{Dataset}} & \multicolumn{3}{c}{\textbf{Gowalla}} & \multicolumn{3}{c}{\textbf{Beauty}} \\
\cmidrule(lr){2-4} \cmidrule(lr){5-7} 
 & \textbf{H@10} & \textbf{N@10} & \textbf{R@10} & \textbf{H@10} & \textbf{N@10} & \textbf{R@10} \\
\midrule
Random-SFM & 0.000018 & 0.000003 & 0.000005 & 0.00085 & 0.000319 & 0.00075 \\
$\delta = 1$ & 0.0217 & 0.0125 & 0.0153 & 0.0472 & 0.0381 & 0.0493 \\
$\delta = 2$ & \textbf{0.0305} & \textbf{0.0148} & \textbf{0.0184} & \textbf{0.0495} & \textbf{0.0435} & \textbf{0.0537} \\
$\delta = 3$ & 0.0264 & 0.0137 & 0.0169 & 0.0491 & 0.0427 & 0.0528 \\
\bottomrule
\end{tabular*}
\end{table}

\subsection{Effectiveness in Positive Pair Alignment (Q5)}
\label{TFP}

To assess the effectiveness of SaFeAU in capturing interacted instances, we compute the average Euclidean distance between normalized representations of positive pairs in the semantic space. Table \ref{tab:alignment_property} presents these distances for the training (\textit{train}) and testing (\textit{test}) sets, while the \textit{Gap} column quantifies the difference between them, calculated as $\textit{Gap} = \textit{test} - \textit{train}$. Based on these results, we draw the following observations:

\begin{itemize}[leftmargin=*, nosep]
\item SaFeAU achieves small average distances on both training and testing sets, indicating that it is capable to learn user and item representations with superior interaction semantics.

\item The shortest distance gaps achieved by SaFeAU demonstrate its strong generalization ability. By aligning semantic latent positive pairs, SaFeAU effectively transfers interaction semantics from the training set to the testing set. Consequently, SaFeAU attains the smallest average distance gaps on both datasets, i.e., 0.0314 for Gowalla and 0.0178 for Toys-and-Games.

\item In contrast, SimGCF achieves a low distance on the training set but a much higher distance on the testing set, indicating that its representations degrade when generalizing to unseen interactions. As the number of graph convolutional layers increases, user and item embeddings tend to collapse toward a common region in the embedding space, reducing the variation among training positive pairs. While this leads to small training distances, the resulting homogeneous representations fail to capture fine-grained interaction patterns, causing poor alignment with positive pairs in the testing set and limiting generalization.

\end{itemize}

\begin{table}[htbp]
\centering
\caption{The average distance of positive pairs.}
\label{tab:alignment_property}
\small 
\setlength{\tabcolsep}{4pt} 
\begin{tabular*}{\columnwidth}{@{\extracolsep{\fill}}lcccccc@{}}
\toprule
\multirow{2}{*}{\textbf{Methods}} & 
\multicolumn{3}{c}{\textbf{Gowalla}} & 
\multicolumn{3}{c}{\textbf{Toys-and-Games}} \\
\cmidrule(lr){2-4} \cmidrule(lr){5-7}
 & \textbf{train} & \textbf{test} & \textbf{Gap} & \textbf{train} & \textbf{test} & \textbf{Gap} \\
\midrule
BPR-MF     & 0.8074 & 0.8458 & 0.0384 & 0.8137 & 0.8576 & 0.0439 \\
LightCCF   & 0.7468 & 0.8097 & 0.0629 & 0.7937 & 0.8194 & 0.0257 \\
SimGCF     & 0.7161 & 0.8336 & 0.1175 & 0.7553 & 0.8413 & 0.0860 \\
DirectAU   & 0.7652 & 0.8222 & 0.0570 & 0.7927 & 0.8329 & 0.0402 \\
SaFeAU     & 0.7529 & 0.7843 & 0.0314 & 0.7864 & 0.8042 & 0.0178 \\
\bottomrule
\end{tabular*}
\end{table}

\subsection{Parameter Sensitivity (Q6)}

We investigated how the five hyperparameters impact the performance of SaFeAU in terms of NDCG@20, including the important weight of uniformity $\gamma_1$, the important weight of semantic positive pair alignment $\gamma_2$, the number of semantic factors $K$, the number of top-ranked semantic factors $k$, the maximum number of iterations $r$.

\begin{figure}[ht]
    \centering
    \includegraphics[width=1.02\linewidth]{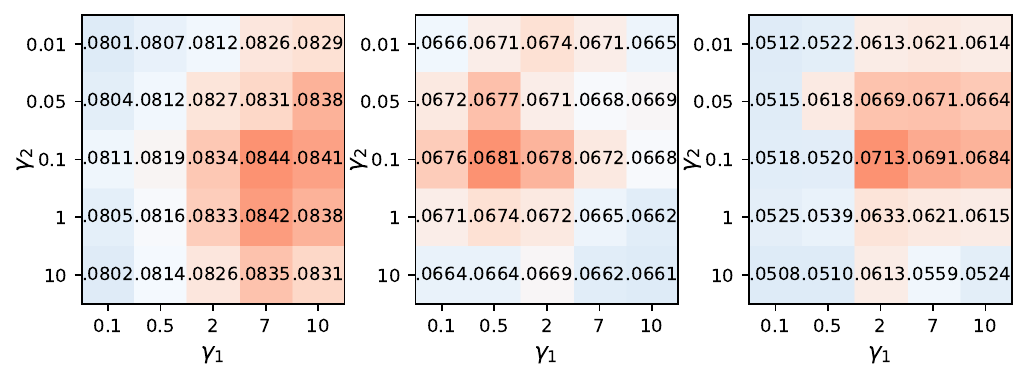}
    
    \vspace{0.4em}
    \makebox[\linewidth][c]{%
        \small \textbf{(a) Gowalla} \hspace{3em}%
        \textbf{(b) Toys-and-Games} \hspace{2em}%
        \textbf{(c) Yelp2018} \hspace{-2.52em}%
    }
    
    \vspace{-0.8em}
    \caption{Impact of the uniformity weight $\gamma_1$ and the semantic positive pairs alignment weight $\gamma_2$ on NDCG@20 across three datasets. \label{fig:3}}
\end{figure}

Figure \ref{fig:3} indicates that, in most cases, the uniformity loss has a greater impact on performance than the semantic positive pairs alignment loss. 
This suggests that, due to sparse interactions, promoting a uniformly distributed representations plays a dominant role, as it enhances the coverage of semantic space and mitigates distribution bias. While larger and denser datasets fit for a higher $\gamma_1$, which means for a stronger uniformity regularization. Smaller and sparser ones like Toys-and-Games achieve peak performance at moderate $\gamma_1$, e.g., $\gamma_1 = 0.5$.

\begin{figure}[ht]
    \centering
    \includegraphics[width=1.02\linewidth]{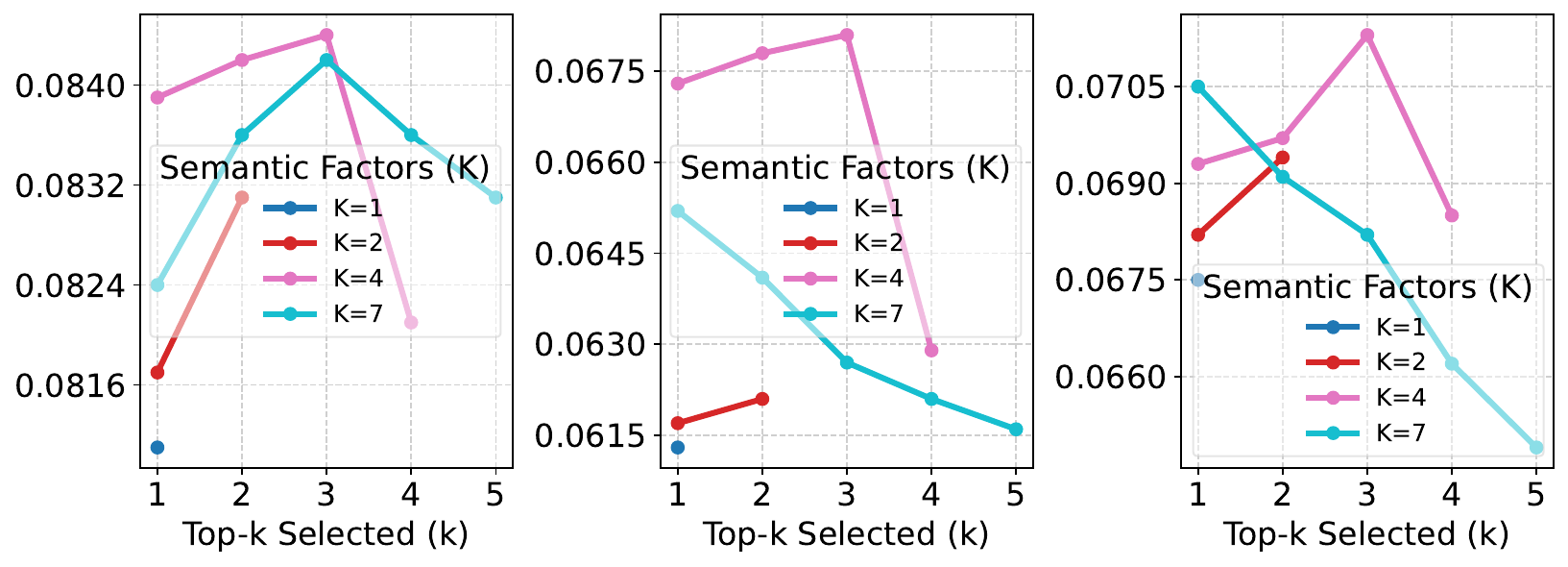}
    
    \vspace{0.4em}
    \makebox[\linewidth][c]{%
        \small \textbf{(a) Gowalla} \hspace{3em}%
        \textbf{(b) Toys-and-Games} \hspace{2em}%
        \textbf{(c) Yelp2018} \hspace{-2.52em}%
    }
    
    \vspace{-0.8em}
    \caption{Effect of the number of semantic factors $K$ and the top-$k$ selection size on NDCG@20 across three datasets.
    \label{fig:4}}
\end{figure}

Figure \ref{fig:4} shows that the best performance is achieved with $K=4$ semantic factors and top-$k=3$ shared factors. Using four factors captures enough semantic diversity among items, while considering only the top-3 shared factors filters out irrelevant information and highlights the most meaningful similarities.

We further investigate the sensitivity of model performance to the number of semantic factor routing iterations $r$. As shown in Figure \ref{fig:routing}, performance across all three datasets plateaus after $r = 4$, with marginal gains from additional iterations. This confirms that $r = 4$ provides a good trade-off between convergence and efficiency.

\vspace{-0.5em}  
\begin{figure}[ht]
    \centering
    \includegraphics[width=0.9\linewidth]{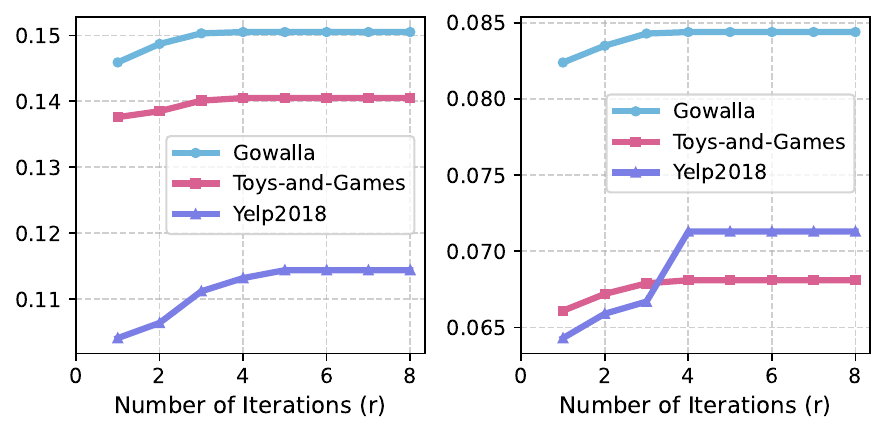} \\
    \vspace{-1.5em}
    
    \subcaptionbox{\label{fig:a}Recall@20}[0.48\linewidth]{\centering\hspace{0.5\linewidth}}%
    \hfill%
    \subcaptionbox{\label{fig:b}NDCG@20}[0.48\linewidth]{\raggedright\hspace{0.5\linewidth}}
    \vspace{-0.6em}
    \caption{The relationship between semantic factor routing iterations and performance across three datasets. \label{fig:routing}}
\end{figure}
\vspace{-0.5em}  

\section{Related Work}

\subsection{Alignment and Uniformity in Recommendation}
Owing to its simplicity and effectiveness, collaborative filtering (CF) is established as a canonical technique in RecSys for helping users discover potential items of interest \cite{LightGCN,NGCF}. Traditional CF methods are mainly based on matrix factorization (MF) \cite{MF1,MF2} and graph convolutional networks (GCNs) \cite{LightGCN,NGCF}. Because GCNs can capture high-order connectivity, they have been widely adopted to learn more powerful user/item representations \cite{NGCF,LightGCN,DISC, FairDDA,MBID, SCKGD}. However, the over-smoothing problem \cite{over-smoothing} remains an inherent challenge in GCNs-based CF.

DirectAU \cite{DirectAU} was the first model to explore the desired properties of representations in CF from the perspective of alignment and uniformity (AU). Building on this work, several subsequent models have been proposed to adopted the AU principle to enhance their predictions. For instance, GraphAU \cite{GraphAU} extends the AU principle by proposing high-order representation alignment, which improves the computational scalability in terms of incorporating multiple graph convolution layers. LightGODE \cite{LightGODE} takes a different approach by questioning the necessity of time-intensive graph convolutions and introducing a highly efficient post-training graph convolution framework. However, the aforementioned methods rely on instance-level learning paradigm and neglect the importance of exploiting false negatives under sparse CF supervision signals. When coping with large scale scenarios, their functionalities are limited to a narrow scope neighborhood alignment due to the computational complexity of GCNs-based methods.

\subsection{Disentangled Representation Learning for Recommendation}
The disentanglement in the feature space encourages the learned representation to carry explainable semantics with independent factors \cite{2013representation}, showing great potential to improve various tasks from the three aspects: interpretability, robustness and controllability \cite{disentangled}. Most existing work on disentangled representation learning have focused on computer vision \cite{smoothing} and natural language processing. In RecSys, users exhibit diverse behaviors and preferences across different items. Applying disentanglement learning to separate these underlying behavioral factors is meaningful for enhancing recommendation quality. However, research in this area is underexplored yet. Most existing approaches only perform coarse-grained disentanglement by separating the target behavior (e.g., buy) from auxiliary behaviors (e.g., click), as in PKEF \cite{PKEF}. Some methods disentangle item representations solely based on the items a single user has interacted with in sequential recommendation \cite{MIND, ComiRec, disentangled}, ignoring factors that are shared across multiple users. Disen-CGCN \cite{Disen-CGCN} considers the need to disentangle multiple factors. However, it mechanically assign feature factors based on embedding dimensions ignoring the fact that, semantic information may not be strictly correlated to a specific embedding dimension in deep learning models.

\section{Conclusion}

In this work, we propose Semantic Factor enhanced Alignment and Uniformity (SaFeAU) to alleviate the impact of false negatives in CF by enhancing supervision signals based on semantic enhanced learning paradigm. In detail, the Semantic Factor Routing algorithm is designed to learn a set of shared semantic factors through an iterative routing process. Then Semantic Factor Matching method is introduced to augment the sparse supervision signals of CF from a semantic perspective. Finally, Semantic Pairs Alignment module leverages the alignment and uniformity principle to align users with both their interacted items and also with SFM-identified potential positive items, while promoting uniformity across all user and item representations for enhancing their disparity. Extensive experiments on four datasets show that SaFeAU consistently outperforms GCN-based and MF-based state-of-the-art CF methods in both recommendation accuracy and computational efficiency. This work presents an attempt to expand supervision signals of CF from the semantic enhanced learning perspective and provides new research possibilities for efficient and large-scale recommendation.

\begin{acks}
This work was supported in part by the National Natural Science Foundation of China (No. 62566015 and No. U22A2099), the Natural Science Foundation of Guangxi Province (No. 2026GXNSFAA00641178), the Innovation Project of GUET Graduate Education (No. 2026YCXS082).
\end{acks}

\clearpage
\bibliographystyle{ACM-Reference-Format}
\balance
\bibliography{reference}
\appendix

\section{Investigation of Loss Functions}

In this section, we first reveal the underlying correlations among three mainstream CF optimization loss, i.e., BPR, InfoNCE and AU. We then show that BPR and InfoNCE rely on limited negative samples and thus incompletely model the negative item distribution. In contrast, AU eliminates negative sampling and directly optimizes the global geometry of the representation space, enabling stable learning even on sparse interactions or large item sets.

\subsection{Relationships Among Loss Functions}
\label{analysis_loss}

In CF, BPR, InfoNCE, and AU are three widely adopted loss functions. However, existing works typically employ these losses in isolation, without thoroughly investigating their deeper correlations in optimization objectives. Fundamentally, all three aim to pull positive samples closer and push negative samples apart in the latent representation space. Equations \eqref{1} and \eqref{2} show that InfoNCE reduces to BPR when \(M = 1\) and \(\tau = 1\). We next investigate the relationship between AU and InfoNCE.

By subtracting $\log M$, we normalize the InfoNCE loss to eliminate its dependence on the number of negative samples. As $M \to \infty$ with fixed $\tau > 0$, the normalized loss converges to a combination of alignment and uniformity terms.

\begin{theorem}[Normalized InfoNCE Converges to AU as $M \to \infty$] For fixed $\tau > 0$, As the number of negative pairs $ M \to \infty $, the normalized InfoNCE loss converges to:
\begin{align}
&\widetilde{L}_{\text{CL}}(f; \tau, M) \notag
=L_{\text{CL}}(f; \tau, M) - \log M \notag\\
&\approx -\frac{1}{\tau} \underset{(u,i) \sim p_{\text{pos}}}{\mathbb{E}} \left[s(u,i)\right] \notag \\
&\quad + \log \underset{(u,u') \sim p_{\text{user}}}{\mathbb{E}} \left[e^{-2 \cdot d(u,u')}\right] + \log \underset{(i,i') \sim p_{\text{item}}}{\mathbb{E}} \left[e^{-2 \cdot d(i,i')}\right]
\end{align}
\end{theorem}

\begin{proof}
\vspace{-0.5em} 
\begin{align}
&\lim_{M \to \infty} \mathcal{L}_{CL} - log M\\
&= \lim_{M \to \infty} - \underset{\substack{(u,i) \sim p_{\text{pos}} \\ \{j_n\}_{n=1}^{M} \overset{\text{i.i.d.}}{\sim} p_{\text{item}}}}{\mathbb{E}} \log \frac{e^{\text{s}(u,i)/\tau}}{e^{\text{s}(u,i)/\tau} + \sum_{n=1}^{M} e^{\text{s}(u,j_n)/\tau}} \notag \\
&- log M\\
\label{7}
&= -\frac{1}{\tau} \underset{(u,i) \sim p_{\text{pos}}}{\mathbb{E}} \left[s(u,i)\right] \notag \\
& +\lim_{M \to \infty} \left[ \underset{\substack{(u,i) \sim p_{\text{pos}} \\ \{j_n\}_{n=1}^{M} \overset{\text{i.i.d.}}{\sim} p_{\text{item}}}}{\mathbb{E}} \log \left( \frac{1}{M} e^{\text{s}(u,i)/\tau} + \frac{1}{M} \sum_{n=1}^{M} e^{\text{s}(u,j_n)/\tau} \right) \right] \\
\label{8}
&= -\frac{1}{\tau} \underset{(u,i) \sim p_{\text{pos}}}{\mathbb{E}} \left[s(u,i)\right] \notag \\
& +\underset{\substack{(u,i) \sim p_{\text{pos}} \\ \{j_n\}_{n=1}^{M} \overset{\text{i.i.d.}}{\sim} p_{\text{item}}}}{\mathbb{E}} \left[ \lim_{M \to \infty} \log \left( \frac{1}{M} e^{\text{s}(u,i)/\tau} + \frac{1}{M} \sum_{n=1}^{M} e^{\text{s}(u,j_n)/\tau} \right) \right] \\
\label{9}
&= -\frac{1}{\tau} \underset{(u,i) \sim p_{\text{pos}}}{\mathbb{E}} \left[s(u,i)\right] + \underset{u \sim p_{\text{user}}}{\mathbb{E}} \left[ \log \underset{j \sim p_{\text{item}}}{\mathbb{E}} \left[e^{s(u,j)/\tau}\right] \right] \\
\label{10}
&= -\frac{1}{\tau} \underset{(u,i) \sim p_{\text{pos}}}{\mathbb{E}} \left[s(u,i)\right] + \underset{u \sim p_{\text{user}}}{\mathbb{E}} \left[ \log ( e^{1/{\tau}} \underset{j \sim p_{\text{item}}}{\mathbb{E}} \left[e^{-d(u,j)/2\tau}\right] )\right]\\
\label{11}
&\approx -\frac{1}{\tau} \underset{(u,i) \sim p_{\text{pos}}}{\mathbb{E}} \left[s(u,i)\right] \notag \\
&\quad + \log \underset{(u,u') \sim p_{\text{user}}}{\mathbb{E}} \left[e^{-2 \cdot d(u,u')}\right] + \log \underset{(i,i') \sim p_{\text{item}}}{\mathbb{E}} \left[e^{-2 \cdot d(i,i')}\right]
\end{align}

We derive from Equation \eqref{7} to Equation \eqref{8} using the Continuous Mapping Theorem and from Equation \eqref{8} to Equation \eqref{9} the strong law of large numbers (SLLN). Since when the user and item embeddings are uniformly distributed on the unit hypersphere, the similarity between a user and an item can be independently estimated based on their respective distances. So, the Equation \eqref{10} is divided by the uniformity of both users and items, utilizing the logarithm of the average pairwise Gaussian potential as described in \cite{AU}.
\end{proof}

This theorem reveals the key insights:
\begin{itemize}[leftmargin=*, nosep]
\item The first term is minimized if and only if the encoder $f$ is perfectly aligned, while the second and third terms are minimized if and only if the user and item embeddings are perfectly uniform.
\item The hyperparameter $\tau$ controls the trade-off between alignment and uniformity.
\end{itemize}

All of this implies that the AU loss implicitly captures the effect of normalized InfoNCE as \(M \to \infty\), thereby achieving the same objective without explicit negative sampling.

\subsection{Limitations of BPR and InfoNCE as Sampling-Based Losses}
\label{M}

Beyond their correlations, BPR and InfoNCE rely on a limited set of negatives from non-interacted items, while AU optimizes the representation geometry over the full item distribution without explicit negatives. Notably, InfoNCE reduces to BPR with one negative and converges to AU as the number of negatives grows, motivating an analysis of the finite-negative approximation error relative to its full-distribution limit.

\begin{theorem}[Approximation Error of Finite-Negative InfoNCE Relative to Its Full-Distribution Limit] 

For fixed $\tau > 0$, the error between the InfoNCE loss using $M$ negative samples and its full-distribution limit decays as:

\begin{align}
&\left| \lim_{M' \to \infty} \widetilde{L}_{\text{CL}}(f; \tau, M') - \widetilde{L}_{\text{CL}}(f; \tau, M) \right| \notag\\
&\leq \frac{1}{M} e^{2/\tau} + O\left(M^{-1/2}\right),
\end{align}
\end{theorem}

\begin{proof}
\vspace{-0.5em} 

\begin{align}
\label{12}
&\left| \lim_{M' \to \infty} (L_{\text{CL}}(f;\tau,M') - \log M') - \left( L_{\text{CL}}(f;\tau,M) - \log M \right) \right|\\
\label{13}
&= \left| \underset{\substack{(u,i) \sim p_{\text{pos}} \\ \{j_n\}_{n=1}^{M} \overset{\text{i.i.d.}}{\sim} p_{\text{item}}}}{\mathbb{E}} \left[ \log \underset{\substack{j \sim p_{\text{item}}}}{\mathbb{E}} \left[ e^{s(u,j)/\tau}\right] - \log(\frac{1}{M} e^{s(u,i)/\tau}+\frac{1}{M} \sum_{n=1}^{M} e^{\text{s}(u,j_n)/\tau} ) \right] \right| \\
\label{14}
&\leq \underset{\substack{(u,i) \sim p_{\text{pos}} \\ \{j_n\}_{n=1}^{M} \overset{\text{i.i.d.}}{\sim} p_{\text{item}}}}{\mathbb{E}} \left[ \left| \left[ \log \underset{\substack{j \sim p_{\text{item}}}}{\mathbb{E}}  e^{s(u,j)/\tau}\right] - \log(\frac{1}{M}e^{s(u,i)/\tau}+\frac{1}{M} \sum_{n=1}^{M} e^{\text{s}(u,j_n)/\tau} ) \right| \right]\\
\label{15}
&\leq e^{1/\tau} \underset{\substack{(u,i) \sim p_{\text{pos}} \\ \{j_n\}_{n=1}^{M} \overset{\text{i.i.d.}}{\sim} p_{\text{item}}}}{\mathbb{E}} \left[ \left| \underset{\substack{j \sim p_{\text{item}}}}{\mathbb{E}}  e^{s(u,j)/\tau} - (\frac{1}{M}e^{s(u,i)/\tau}+\frac{1}{M} \sum_{n=1}^{M} e^{\text{s}(u,j_n)/\tau} ) \right| \right] \\
\label{16}
&\leq \frac{1}{M} e^{2/\tau} + e^{1/\tau} \underset{u,\{j_n\}_{n=1}^{M} \overset{\text{i.i.d.}}{\sim} p_{\text{item}}}{\mathbb{E}} \left[ \underset{\substack{j \sim p_{\text{item}}}}{\mathbb{E}} [ e^{s(u,j)/\tau} ] - \frac{1}{M} \sum_{n=1}^{M} e^{\text{s}(u,j_n)/\tau} \right] \\
\label{17}
&= \frac{1}{M} e^{2/\tau} + O\left(M^{-1/2}\right),
\end{align}

The inequality from Equation \eqref{13} to Equation \eqref{14} follows from Jensen's inequality applied to the absolute value function. The inequality from Equation \eqref{14} to Equation \eqref{15} follows the Intermediate Value Theorem and the $e^{1/\tau}$ upper bound on the absolute derivative of log between the two points. For Equation \eqref{15}, we apply the triangle inequality to separate the positive sample term, as shown in Equation \eqref{16}. Applying the Berry-Esseen theorem to the i.i.d. negative sample terms with bounded support, we establish that the mean estimation error scales as $O(M^{-1/2})$.
\end{proof}

By analyzing the results, we can draw the following conclusions:
\begin{itemize}[leftmargin=*, nosep]
\item \textbf{BPR vs. InfoNCE:} BPR underperforms InfoNCE with larger $M$ because it corresponds to the extreme case of $M=1$, which incurs the largest approximation error. 

\item \textbf{Negative Sampling Error Bottleneck:} The approximation error has two components, $O(M^{-1})$ from positive-pair normalization and $O(M^{-1/2})$ from finite negative sampling. While the former vanishes quickly, the latter decays more slowly and dominates the total error, limiting how closely the finite-sample InfoNCE loss can approach its full-distribution limit.

\item \textbf{InfoNCE vs. AU:} In the sparse interaction scenarios, where negative sampling suffers from false negative labeling and a narrow view of the negative space, the AU objective is more robust, as it eliminates the need for negative sampling and directly optimizes alignment and uniformity over the global representation space. 
\end{itemize}

Unlike BPR and InfoNCE, which rely on negative sampling, the AU objective directly optimizes the global geometry of the representation space, making it more robust. We adopt AU as our base loss to reduce sampling-induced biases and improve generalization in sparse and large-scale recommendation scenarios.

\section{Theoretical Analysis of Semantic-Aware Alignment Enhancement}

In the previous section, we discussed the relationship between AU and InfoNCE/BPR, and showed that AU avoids explicit negative sampling by directly optimizing alignment and uniformity. However, under sparse implicit feedback, the alignment supervision in AU is still limited to observed user-item interactions. To address this limitation, we further analyze how SFR decomposes item embeddings into semantic factors for semantic relation modeling, while SPA enriches the alignment supervision by augmenting supervised signals.

\subsection{Theoretical Analysis of Semantic Factor Routing}
\begin{theorem}[Semantic Factor Routing for Soft Decomposition]
Our semantic factor routing algorithm optimizes the following objective function. After convergence, the learned semantic factors can be viewed
as a soft semantic decomposition of the input item embeddings.

\begin{equation}
\begin{aligned}
    \min_{W,\mathcal{F}}  \mathcal{L}(W,\mathcal{F}) &:= - \sum_j \langle z_{j}, \mathcal{F}_j \rangle + \alpha \sum_i \sum_j w_{ij} \log w_{ij}, \\
    \text{s.t.} \quad & \sum_{j} w_{ij} = 1, \, w_{ij} > 0, \, \|\mathcal{F}_j\| \leq 1
\label{18}
\end{aligned}
\end{equation}
\end{theorem}

\noindent \textbf{Intuition behind the objective}. Equation~\eqref{18} provides an optimization interpretation of Algorithm~\ref{alg:SFR}. The term $-\sum_j \langle z_j, \mathcal{F}_j \rangle$ encourages routed item embeddings to align with their corresponding semantic factors. Accordingly, the routing weights $w_{ij}$ are updated through a softmax operation over factor-item similarities in Algorithm~\ref{alg:SFR}. The regularization term $\alpha \sum_i \sum_j w_{ij}\log w_{ij}$ acts as an entropy regularizer, which prevents one-hot routing assignments and leads to softer routing distributions across semantic factors.

\begin{proof}
We optimize Equation~\eqref{18} by alternating minimization over
the routing weights $W$ and semantic factors $\{\mathcal F_j\}$.

\noindent\textbf{(1) Optimization with respect to $W$.} Fixing $\{\mathcal F_j\}$ and substituting $z_j=\sum_i w_{ij}\tilde f(i)$, the objective becomes
\begin{equation*}
\min_W -\sum_i\sum_j w_{ij} \langle \tilde f(i), \mathcal F_j \rangle + \alpha\sum_i\sum_j w_{ij}\log w_{ij},
\end{equation*}
subject to
\begin{equation*}
\sum_j w_{ij}=1, \qquad w_{ij}>0.
\end{equation*}

Introducing the Lagrange multiplier $\lambda_i$ for the simplex constraint yields
\begin{equation*}
-\sum_j w_{ij}\langle \tilde f(i),\mathcal F_j\rangle + \alpha\sum_j w_{ij}\log w_{ij} + \lambda_i \left( \sum_j w_{ij} - 1 \right).
\end{equation*}

Taking derivatives with respect to $w_{ij}$ gives
\begin{equation*}
-\langle \tilde f(i),\mathcal F_j\rangle + \alpha(\log w_{ij}+1) + \lambda_i = 0.
\end{equation*}

Solving for $w_{ij}$ yields
\begin{equation*}
w_{ij} \propto \exp \left( \frac{ \langle \tilde f(i),\mathcal F_j\rangle }{\alpha} \right),
\end{equation*}
which corresponds to a softmax routing distribution over semantic factors.

\noindent \textbf{(2) Optimization with respect to $\mathcal F_j$.} Fixing $W$, the optimization problem reduces to
\begin{equation*}
\max_{\|\mathcal F_j\|\le1}
\langle z_j,\mathcal F_j\rangle.
\end{equation*}

By the Cauchy--Schwarz inequality,
\begin{equation*}
\langle z_j,\mathcal F_j\rangle \le \|z_j\|\|\mathcal F_j\| \le \|z_j\|,
\end{equation*}
where equality holds when
\begin{equation*}
\mathcal F_j = \frac{z_j}{\|z_j\|}.
\end{equation*}

Therefore, each semantic factor corresponds to the normalized direction of the weighted aggregation of routed item embeddings.

\noindent \textbf{(3) Convergence.} Each step minimizes Equation~\eqref{18} with one variable block fixed, thus the objective monotonically decreases during optimization. Since the objective is lower bounded, the optimization converges to a stationary point, yielding a soft-routing decomposition of the input item embeddings.
\end{proof}

\subsection{Generalization Benefit of Positive Augmentation}

\begin{theorem}[Generalization Benefit of Positive Augmentation]

Recall that the alignment objective is defined as:
\begin{equation*}
L_{\text{align}}=\underset{(u,i)\sim p_{\text{pos}}}{\mathbb E}\left\|\tilde f(u)-\tilde f(i)\right\|^2.
\end{equation*}
Given $N$ positive pairs
$\{(u_n,i_n)\}_{n=1}^N$
sampled i.i.d. from $p_{\text{pos}}(u,i)$, the empirical alignment loss is:
\begin{equation*}
\hat L_{\text{align}}=\frac1N\sum_{n=1}^N\|\tilde f(u_n)-\tilde f(i_n)\|^2.
\end{equation*}
Then, with probability $\ge1-\delta$, the estimation error is bounded by:
\begin{equation}
\left|L_{\text{align}}-\hat L_{\text{align}}\right|\le\mathcal{O}(N^{-1/2}).
\end{equation}
\end{theorem}

\begin{proof}
Define
\begin{equation*}
X_n=\|\tilde{f(u_n)}-\tilde{f(i_n)}\|^2.
\end{equation*}

The empirical alignment loss can be written as:
\begin{equation*}
\hat L_{\text{align}}=\frac{1}{N}\sum_{n=1}^{N} X_n,
\end{equation*}
while the population alignment loss is:
\begin{equation*}
L_{\text{align}}=\mathbb{E}[X].
\end{equation*}

Since embeddings are normalized, we have $X_n \in [0,4]$. Applying Hoeffding's inequality, the probability that the empirical alignment loss deviates from the population alignment loss by more than $t$ is bounded by:
\begin{equation*}
\Pr\left(\left|\hat L_{\text{align}} - L_{\text{align}}
\right|\ge t\right)\le2\exp\left(-\frac{2Nt^2}{16}\right).
\end{equation*}

Equivalently, for any failure probability $\delta > 0$, solving
\begin{equation*}
2\exp\left(-\frac{2Nt^2}{16}\right)=\delta,
\end{equation*}
gives
\begin{equation*}
t=4\sqrt{\frac{\ln(2/\delta)}{2N}}.
\end{equation*}

Therefore, with probability at least $1-\delta$, we obtain:
\begin{equation*}
\left|\hat L_{\text{align}}-L_{\text{align}}\right|\le4\sqrt{\frac{\ln(2/\delta)}{2N}}=\mathcal{O}(N^{-1/2}).
\end{equation*}

The estimation error decreases at rate $\mathcal{O}(N^{-1/2})$ as the number of positive pairs increases. Thus, SPA improves alignment optimization by augmenting additional semantic-aware positive pairs.
\end{proof}

\section{Experimental Setup}

\subsection{Baselines} \label{baselines}
\begin{itemize}[leftmargin=*, nosep]
\item \textbf{BPR-MF} \cite{BPR} is a learning-to-rank method that improves a model's ability to rank preferred items above negative ones.
\item \textbf{DirectAU} \cite{DirectAU} proposes a new learning objective for CF that directly optimizes the alignment and uniformity of user and item representations on a hypersphere. 
\item \textbf{LightGODE} \cite{LightGODE} replaces the computationally expensive GCNs during training with an efficient post-training graph ODE.
\item \textbf{LightCCF} \cite{LightCCF} reveals that the contrastive learning objective inherently performs GCN through its gradient descent process.
\item \textbf{LightGCN} \cite{LightGCN} is a simplified GCN for CF that performs linear propagation between neighbors on the user-item bipartite graph.
\item \textbf{GraphAU} \cite{GraphAU} is a graph-based method that improves hypersphere representation learning for RecSys by using high-order connections in the user-item graph to alleviate data sparsity.
\item \textbf{FourierKAN-GCF} \cite{FourierKAN-GCF} introduces FourierKAN as an efficient feature transformation module in graph collaborative filtering, thereby enhancing the model’s representation capability while reducing training difficulty.
\item \textbf{SimGCF} \cite{KDD26} enhances GCNs by adaptively scaling low- and high-frequency graph signals and applying a space flip operation to recover signal expressiveness.
\end{itemize}

\begin{table}[t]
\centering
\caption{Training time comparison of GCN-based and MF-based methods on three datasets, including the average epoch time, number of epochs, and total training time. Here, s, m, and h denote seconds, minutes, and hours, respectively.}
\label{tab:training_time_comparison}
\small
\setlength{\tabcolsep}{4pt}
\begin{tabular}{@{}lccccc@{}}
\toprule
Dataset & Method & Time/Epoch & Epochs & Total Time \\
\midrule
\multirow{6}{*}[-3pt]{\centering\footnotesize Gowalla}
 & LightGCN & 112.32s & 113 & 3.53h \\
 & FourierKAN-GCF & 126.58s & 91 & 3.20h \\
 & SimGCF & 245.29s & 87 & 5.93h \\
 & LightGODE & 14.67s & 94 & 0.38h \\
 & LightCCF & 14.05s & 88 & 0.34h \\
 & SaFeAU & 20.87s & 58 &  0.34h\\
\midrule
\multirow{6}{*}[-3pt]{\centering\footnotesize Toys-and-Games}
 & LightGCN & 1.82s & 134 & 4.065m  \\
 & FourierKAN-GCF & 2.56s & 122 & 5.205m \\
 & SimGCF & 2.81s & 118 & 5.526m \\
 & LightGODE & 0.52s & 130 & 1.13m \\
 & LightCCF & 0.54s & 164 & 1.48m \\
 & SaFeAU & 0.73s & 98 & 1.19m \\
\midrule
\multirow{6}{*}[-3pt]{\centering\footnotesize Yelp2018}
 & LightGCN & 49.52s & 59 & 48.70m \\
 & FourierKAN-GCF & 50.86s & 49 & 41.54m \\
 & SimGCF & 64.18s & 42 & 44.93m \\
 & LightGODE & 1.73s & 113 & 3.26m \\
 & LightCCF & 1.88s & 109 & 3.42m \\
 & SaFeAU & 2.32s & 84 & 3.25m \\
\bottomrule
\end{tabular}
\end{table}

\subsection{Efficiency Analysis} 
\label{efficiency}

We evaluate the computational efficiency of SaFeAU by comparing its training time with state-of-the-art methods, including traditional GCN-based methods (LightGCN, FourierKAN-GCF, SimGCF) and efficient MF-based methods (LightGODE, LightCCF). Results on three representative datasets are summarized in Table \ref{tab:training_time_comparison}, from which several key observations emerge:

\begin{itemize}[leftmargin=*, nosep]
\item MF-based models (LightGODE, LightCCF, SaFeAU) achieve significantly faster training than GCN-based models (LightGCN, SimGCF) by avoiding costly neighborhood aggregation.

\item SaFeAU maintains high training efficiency as data complexity grows, with only a modest increase in per-epoch runtime, making it suitable for large-scale industrial RecSys.

\item SaFeAU also converges in fewer epochs, indicating that its semantic positive pairs alignment loss captures user preferences more effectively and achieves peak performance in fewer training iterations.
\end{itemize}
\end{document}